\documentclass[10pt, conference]{IEEEtran}
\IEEEoverridecommandlockouts
\usepackage{cite}
\usepackage{textcomp}
\usepackage[a4paper]{geometry}
\usepackage{marvosym}

\usepackage{CJKutf8}
\usepackage{changes}

\usepackage{gensymb}
\usepackage{bm}

\usepackage{amsmath,amssymb}
\usepackage{amsthm}

\usepackage{balance}

\usepackage{graphicx}
\graphicspath{{figures/}}
\usepackage[export]{adjustbox}
\usepackage{flushend}

\usepackage{subfigure}

\usepackage{listings}
\usepackage{color}
\usepackage{soul}
\usepackage{tikz}
\newcommand*{\circled}[1]{\lower.7ex\hbox{\tikz\draw (0pt, 0pt)%
    circle (.5em) node {\makebox[1em][c]{\small #1}};}}

\usepackage{algorithm}
\usepackage{algpseudocode}

\theoremstyle{definition}
\newtheorem{definition}{Definition}
\newtheorem{example}{Example}
\newtheorem{theorem}{Theorem}

\newtheorem*{proof}{Proof}
 
\theoremstyle{remark}

\definecolor{mygreen}{rgb}{0,0.6,0}
\definecolor{mygray}{rgb}{0.5,0.5,0.5}
\definecolor{mymauve}{rgb}{0.58,0,0.82}

\usepackage{algpseudocode}
\usepackage{algorithm}

\usepackage{booktabs} 
\usepackage{multirow}
\usepackage{array}
\newcolumntype{I}{!{\vrule width 1.2pt}}
\newlength\savedwidth

\newlength\savewidth

\usepackage{makecell}

\usepackage{etoolbox}
\makeatletter
\patchcmd{\@makecaption}
  {\scshape}
  {}
  {}
  {}
\makeatletter
\patchcmd{\@makecaption}
  {\\}
  {.\ }
  {}
  {}
\makeatother

\usepackage{geometry}
\begin{document}
\begin{CJK*}{UTF8}{gbsn}

\title{MEC: An Open-source Fine-grained \textbf{M}apping \textbf{E}quivalence \textbf{C}hecking Tool for FPGA}

\author{
\IEEEauthorblockN{
Liwei~Ni$^{1,2}$\IEEEauthorrefmark{1}, Zonglin~Yang${^{5,1}}$, Jiaxi~Zhang${^{3}}$, Changhong~Feng${^{4}}$, Jianhua~Liu${^{4}}$, \\ Guojie~Luo${^{3,1}}$, Huawei~Li${^{2,1}}$, Biwei~Xie$^{2,1}$ and~Xingquan~Li$^{1,\text{\Letter}}$
}
\IEEEauthorblockA{
$^1$Peng Cheng Laboratory, Shenzhen, China
}
\IEEEauthorblockA{
$^2$Institute of Computing Technology, Chinese Academy of Sciences, Beijing, China
}
\IEEEauthorblockA{
$^3$Center for Energy-Efficient Computing and Applications, Peking University, Beijing, China
}

\IEEEauthorblockA{
$^4$Shanghai Anlogic Infotech Co., Ltd, Shanghai, China
}

\IEEEauthorblockA{
$^5$Shenzhen University, Shenzhen, China
}

\IEEEauthorblockA{
Email: \IEEEauthorrefmark{1}nlwmode@gmail.com, $^{\text{\Letter}}$lixq01@pcl.ac.cn
}}

\maketitle

\begin{abstract}


Technology mapping is an essential step in EDA flow. 
However, the function of the circuit may be changed after technology mapping, and equivalence checking (EC) based verification is highly necessary. 
The traditional EC method has significant time and resource constraints, making it only feasible to carry out at a coarse-grained level. 
To make it efficient for technology mapping, we propose a fine-grained method called MEC, which leverages a combination of two approaches to significantly reduce the time cost of verification. 
The local block verification approach performs fast verification and the global graph cover approach guarantees correctness. 
The proposed method is rigorously tested and compared to three EC tools, and the results show that MEC technique offers a substantial improvement in speed. 
MEC not only offers a faster and more efficient way of performing EC on technology mapping but also opens up new opportunities for more fine-grained verification in the future.

\end{abstract}


\section{Introduction}



Technology mapping, an essential process in logic synthesis, transforms a technology-independent logic network into a network with logic blocks. In FPGA technology mapping\cite{techmap-c}, the logic network is transformed into a $k$-LUT network, where each $k$-LUT block can represent any Boolean function with an input size less than or equal to k, due to the reconfigurability of the LUTs which are similar to SRAMs. Making the verification process of technology mapping more cost-effective and having a debug tool available are important considerations to improve the efficiency and accuracy of the process.


Equivalence Checking is also an important step of Electronic Design Automation~(EDA) flow and is commonly used to prove that two representations of a digital circuit design exhibit exactly the same behavior. It can also guarantee the correctness of technology mapping. As illustrated in Fig. \ref{fig:formal}, EC is usually performed after each step to ensure that the behavior has not changed during the current step. And the coarse-grained EC directly performs verification on the output of two circuits. With the growth in the size of Very Large Scale Integrated (VLSI) Circuits, the time spent on the verification step has increased to almost $70\%$ \cite{pdv-book} in IC-flow.

\begin{figure}[tbp]
    \centering
    \includegraphics[scale=0.36]{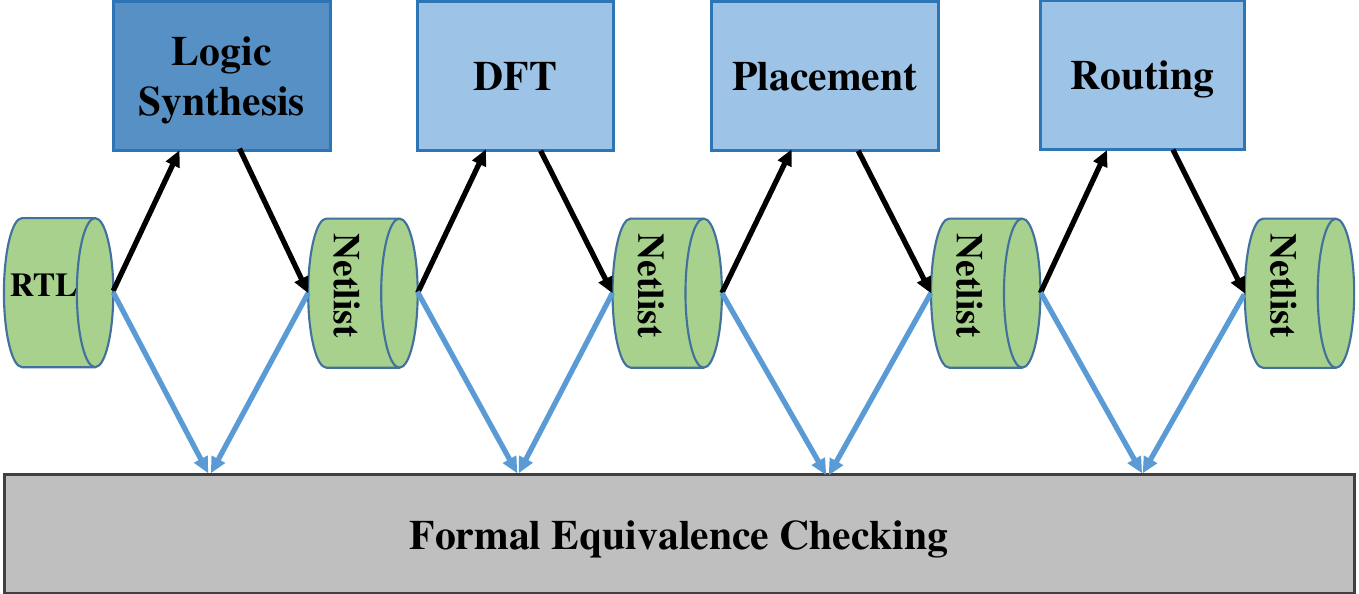}
    \caption{Verification for EDA flow.}
    \label{fig:formal}
\end{figure}

Reducing the time of the EDA flow can be achieved by making the verification process more efficient, and there have been numerous efforts in this direction.
Goldberg et al. \cite{sat_ec-c} introduced SAT to solve the equivalence checking problem and showed that it is a more robust and flexible engine for Boolean reasoning than BDD-based methods. And this led to the widespread use of SAT-based equivalence checking. 
Cook showed   that the SAT problem is NP-complete  in \cite{sat_cook-c}, triggering further research aimed at reducing the complexity of the problem for SAT solvers. This also means that cost time can be reduced by reducing the size of the circuit. Alan et al.~\cite{abc_cec-c} mainly reduced the cost by using more intelligent simulation and introducing SAT-sweep to find function-equivalent node-pairs to reduce the solution space, further enhancing the SAT-based method. 
There are also other approaches to speed up verification. 
Feng et al.~\cite{points_insert-c} reduced the checking scale by early inserting cutpoints into the RTL, allowing for verification through these cutpoints, but this approach carries the risk of changing the design.

It is possible to conduct fine-grained verification if the verification method can be made more efficient. This paper focuses on the equivalence checking of FPGA technology mapping in logic synthesis. Inspired by~\cite{lut_sat-c}, which applies FPGA technology mapping to speed up SAT solver, we can make a fast verification by incrementally checking the equivalence of the mapping result network and its original network in topological order. 

Our contributions can be divided into three parts:
\begin{itemize}
    \item We prove that the verification on technology mapping can be based on local-block-verification and global-graph-cover, and propose a highly efficient method for EC on technology mapping.
    \item We release a debug tool for technology mapping, \textbf{MEC}\footnote{https://github.com/nlwmode/MEC.}, and it can also expose the detailed equivalence information for the local block to make a quick debug.
    \item We compare MEC with three EC tools, the results of the experiment show that MEC offers a substantial improvement in speed.
\end{itemize}

The rest of this paper is organized as follows:
Section II provides background information on technology mapping and SAT-based equivalence checking.
Section III gives an illustration of the details of our proposed tool, MEC.
Section IV  shows the framework of our work.
Section V  presents the results of our experiments on equivalence checking during technology mapping.
Finally,  conclusion and future work are summarized in Section VI.

\section{Preliminaries}








In this section, we consider the verification on FPGA technology mapping. And we will introduce some of the basics we might need later.

A Boolean Network is a \textit{directed acyclic graph} (DAG) which the nodes corresponding to logic gates and the edges corresponding to wires between gates. 
A node $n$ has zero or more \textit{fanins} and \textit{fanouts}, fanins mean the nodes that flowed in and fanouts mean the nodes that flowed out. The \textit{primary inputs}~(PIs) are nodes without fanins and the \textit{primary outputs}~(POs) are nodes without fanouts. A \textit{fanin/fanout cone} of node $n$ is the subset nodes of the network reachable through the fanin/fanout edges from $n$.

As for choice-related computation, we need to extend the traditional DAG with additional properties of displayed function-equivalent or function-complementary edges between nodes. The function-equivalent edges refer to the bi-directional solid arrow line and the function-complementary edges refer to the bi-directional dotted arrow lines.
A \textit{representative node} of a network is the node without any displayed function-equivalent or function-complement edges in a network. 
The \textit{choice class} is the set of nodes that are displayed function-equivalent or function-complementary in a network. And only one representative node is needed for this class to connect all its fanouts. 
The \textit{choice node} is the node that is not the representative node in a choice class.
The \textit{regular network} is the network that all the nodes are representative nodes, and it does not have any choice class.
The \textit{choice network} is the network with some choice classes. In a choice network, nodes in a choice class are connected by a linked list and are connected by a bi-directional solid arrow if the choice node is function-equivalent to its representative node. If not, the arrow is bi-directional dotted.

For cut-based technology mapping, a \textit{cut} $c$ of node $n$ is a subset of nodes of transitive fanin cone that each path from a PI to $n$ passes through at least one node of $c$, in short, $n$ is the root, and $c$ are the leaves. A \textit{trivial cut} of node $n$ is the node itself. A \textit{non-trivial cut} includes all the nodes that are found on the paths from the root to leaves. A $k$-\textit{bounded network} is a network where the number of fanins of any node that does not exceed $k$, and also a cut is $k$-\textit{feasible} if the number of nodes in it does not exceed $k$.

A \textit{mapping} of a network means that each non-PI node assigned a $k$-feasible cut also called \textit{representative cut}. So we can easily get the mapping result for each non-PI node's representative cut from POs to PIs incrementally, and the reached nodes are used, the covered nodes are hidden in a cut as an internal part of logic block.

A \textit{miter} $m$ of two Boolean networks $f$ and $g$, $f$ and $g$ have the same PIs and POs, is the merged network that adding XOR gates to the same corresponding outputs for $f$ and $g$, then adding a big OR gate for these XOR gates. It is easy to do equivalence checking by the miter network.

\section{Traditional technology mapping and verification}
In this section, we will outline the basic flow for FPGA technology mapping and provide details on SAT-based equivalence checking.

\subsection{Traditional FPGA Technology Mapping}
Algorithm \ref{algorithm:map} depicts the traditional FPGA technology mapping flow as implemented in ABC \cite{abc-tool}.  In the following subsections, we will delve into the details of each step of the algorithm process. \\

\begin{algorithm}
\small
\caption{Traditional FPGA technology mapping}
\label{algorithm:map}
\begin{algorithmic}[1]
\Require Network $G$, cut size $k$
\Ensure Mapping result $mr$

\Statex /* merge the internal optimized networks to reduce bias */ 
\State \textit{CG} $\leftarrow$ compute\_choice\_network($G$)
\Statex /* compute all $k$-feasible cuts at each node */ 
\State cut\_enumeration(\textit{CG}, $k$)
\Statex /* compute min-depth representative cut for each node */
\State depth\_oriented\_mapping(\textit{CG}, $k$)
\Statex /* update representative cut at each node to optimize area */
\State area\_recovery(\textit{CG}, $k$)
\Statex /* compute the set of nodes used in the final mapping */
\State $mr$ $\leftarrow$ derive\_final\_mapping(\textit{CG}, $k$)
\State \Return $mr$
\end{algorithmic}
\end{algorithm}






\subsubsection{Choice network generation}
\ \newline 
\indent Due to the bias\cite{bias_reduction-c} of technology-independent optimizations, meaning that the lowest-cost network may not obtain the best mapping result, we can reduce the bias by merging the internal optimized networks into a choice network that retains the local structure of each selected network. In this way, we can obtain a better mapping result compared to each individual network.



\subsubsection{Cut enumeration}
\ \newline 
\indent Cut enumeration \cite{cut_rank-c} serves as the foundation for cut-based optimization and technology mapping. Given two cut sets, $CS_1$ and $CS_2$, the merge operation of the cut sets, $CS_1 \diamond CS_2$, is defined as:
$$
CS_1 \diamond CS_2 = \{u \cup v | u \in CS_1, v \in CS_2\}
\eqno{(1)}
$$

As for technology mapping, we also need to add one constraint to equation (1): $\{|u\cup v| \leq k\}$ which means that the new merged cut's size is no more than k. And we can define the equation to compute each node's $k$-feasible cut: for convenience, let $\Phi(n)$ be the set of $k$-feasible cuts for two-input node $n$, $n_1$ and $n_2$ denote to $n$'s fanins:  

$$
\Phi(n) = \left\{
\begin{aligned}
\{\{n\}\} & :n\in PIs \\
\{\{n\}\} \cup \Phi(n_1) \diamond \Phi(n_2)  & :otherwise
\end{aligned}
\right \}.
\eqno{(2)}
$$

To begin with, each node will be assigned a trivial cut consisting of itself. 
Then, the cut set for each node will be generated incrementally in topological order by merging the cut sets of its children, according to equation (2).
To balance the trade-off between memory, time, and mapping result, the priority cut algorithm \cite{priority_cuts07-c} is usually adopted, which only saves better cut sets according to some sorting criteria for optimization or technology mapping.




\subsubsection{Derive mapping result}
\ \newline 
\indent Depth-oriented mapping provides an initial representative cut assignment for each node to achieve a depth-optimal mapping result. 
However, there is still room for improvement in terms of area. 
The area recovery step involves finding a more area-efficient cut to be used as the representative cut, while maintaining the constraint of depth. 

After the steps above are completed, we can extract the mapping result by traversing the network in reverse topological order and collecting the representative cuts of the used nodes. This step ensures that the entire network is covered. It should be noted that trivial cuts can not be representative cuts for internal nodes, as we must ensure full coverage of the network. 


\subsection{SAT-based Equivalence Checking}

Fig.\ref{fig:ec} illustrates the scheme for SAT-based equivalence checking. Due to the property of XOR gates, the equivalence of $f$ and $g$ can easily be checked by checking the satisfiability of the miter network. In detail, the process involves first transforming the miter network into CNF form using the Tseitin transformation \cite{tseytin-book} and then checking it with SAT solver.

\vspace{3pt}
\begin{figure}[tbp]
    \centering
    \includegraphics[scale=0.25]{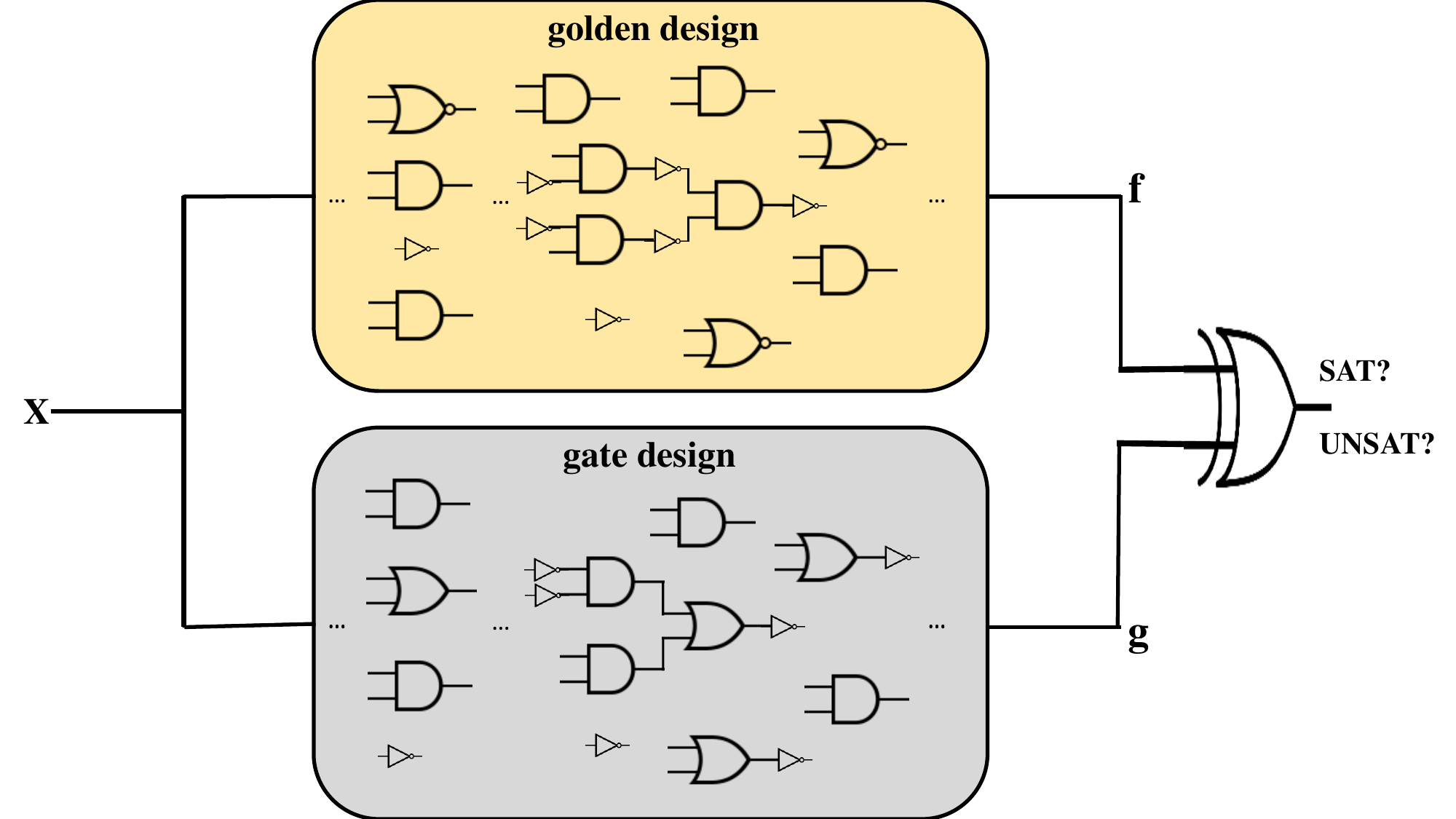}
    \caption{The miter network for $f$ and $g$.}
    \label{fig:ec}
\end{figure}

We can say that $f$ and $g$ are equivalent only if the miter network is UNSAT, meaning that $f$ and $g$ have the same output for any input signals. On the other hand, if the miter network is SAT, it means that there are some input signals that result in different outputs for $f$ and $g$, and therefore $f$ and $g$ are non-equivalent.



\section{Verification on technology mapping}

\begin{figure*}[tbp]
    \centering
    \includegraphics[scale=0.31]{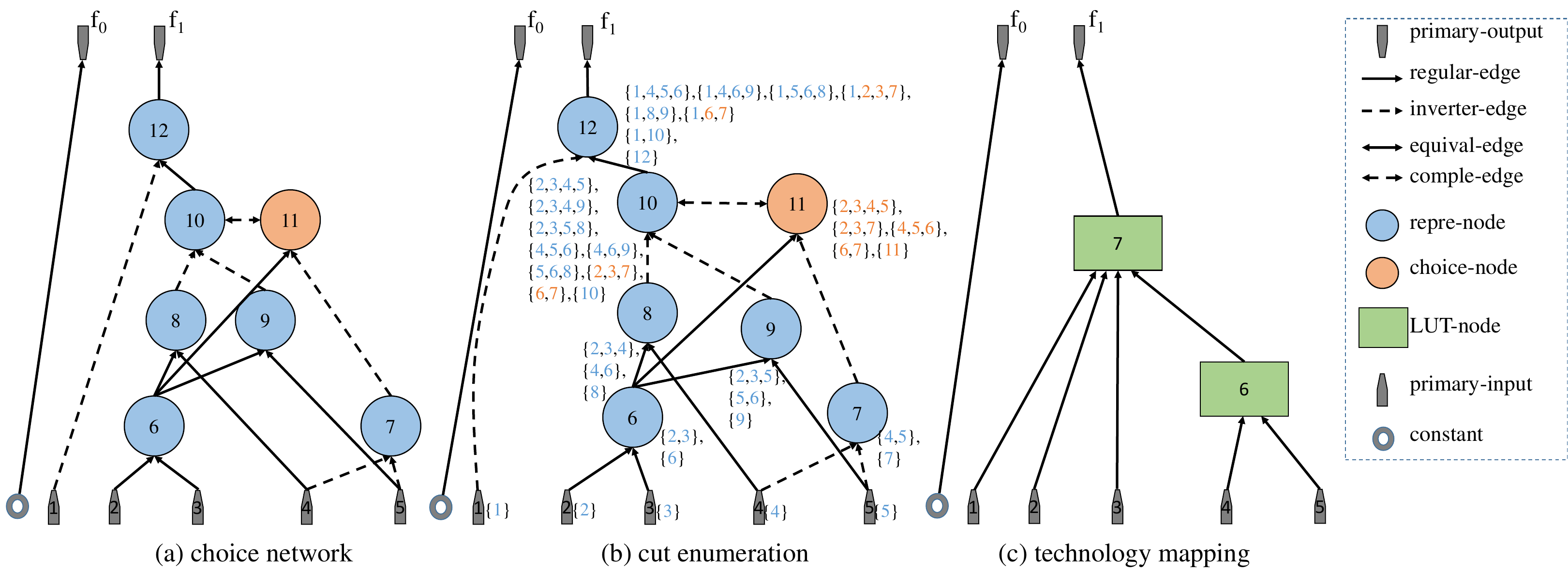}
    \caption{The technology mapping on choice network.}
    \label{fig:map}
\end{figure*}

In this section, we will illustrate the details of our proposed algorithm for verification on FPGA technology mapping, MEC. We will also prove that the results obtained from our algorithm are equivalent to those obtained from the miter network.

Our core idea is to reduce the size of the problem, and we can optimize the checking problem by combining local block verification and global graph cover. Miter-based equivalence checking for the two networks is time-consuming for technology mapping and therefore not practical for verification here. However, our proposed algorithm, logic block-based equivalence checking, can produce the same results in much less time and also provides local structure information for debugging technology mapping.
 
\subsection{Equivalence Checking on Technology Mapping}

Algorithm \ref{algorithm:ec} displays the complete flow of our proposed algorithm. We emphasize that our proposed algorithm can be adopted to combinational network and sequential network. Theorem \ref{thm-ec} also confirms the correctness of our method's checking results. 

\begin{definition}
\label{def-INS}
(\textit{greatest node set}~(\textit{GNS})). Let $N$ be a set of nodes of original network (\textit{OG}), and $\preccurlyeq$ be the partial ordering relation referred to the directed edge. If $(N,~\preccurlyeq)$ is a partial ordering set with only a greatest element $ge\in N$. We call above greatest element $ge$ as greatest node, and all greatest nodes compose greatest node set (\textit{GNS}).
\end{definition}

\begin{theorem}
\label{thm-ec}
A mapping network (\textit{MG}) equivalents to its original network (\textit{OG}) if each logic block in \textit{MG} equivalents to its covered network.
\end{theorem}

\begin{proof}
A mapping is a set of representative cuts for used nodes in the network, and the used nodes are the \textit{GNS} in \textit{OG}. And the corresponding nodes \textit{CNS} to \textit{GNS} in \textit{MG} can be easily obtained during deriving mapping. \textit{GNS} guarantees the global cover and its internal nodes are only connecting fanin cone and fanout cone, and have no logic function. Therefore, we can get the verification result of \textit{MG} and \textit{OG} by checking the equivalence based on \textit{GNS} and \textit{CNS} in topological order.
\end{proof}

\begin{definition}
\label{def-ICS}
(\textit{corresponding pair set}~(\textit{CPS})). For a node pair $(i,~j)$, if $i \in $\textit{GNS} in \textit{OG}, and  $j$ is its corresponding node in \textit{MG}, we call this node pair as corresponding pair, and \textit{CPS} is the set of corresponding pair.
\end{definition}

\vspace{3pt}
\begin{algorithm}
\small
\caption{Verification on Technology Mapping}
\label{algorithm:ec}
\begin{algorithmic}[1]
\Require  network \textit{OG} before mapping, 

mapping result network \textit{MG}
\Ensure Equivalent or not
\State initialize $error$ = false
\State collecting the \textit{CPS} for \textit{OG} and \textit{MG}
\State topologizing \textit{CPS}
\For{$np$ in \textit{CPS}}
    \If{$error$ == true}
        \State \Return !$error$
    \Else
        \State $error$ = $checking\_logic\_block$($np$, \textit{OG}, \textit{MG})
    \EndIf
\EndFor
\State \Return !$error$
\end{algorithmic}
\end{algorithm}

\subsubsection{Equivalence checking on combinational mapping}
\ \newline 
\indent As illustrated in Algorithm \ref{algorithm:ec}, after technology mapping, we obtain the mapping network \textit{MG} for its original network \textit{OG}. 
The variable $error$ indicates the presence of bugs in the current mapping network. 
First, we collect the injective corresponding node set \textit{CPS} which can expose the corresponding node-pairs between \textit{OG} and \textit{MG}. 
The processing order is important because it is meaningless to check for equivalence after node-pair $np$ if we find that its logic block is not equivalent to its covered network. 
After ordering the \textit{CPS} in a topological manner, we perform the equivalence checking incrementally. 
If we find that a node's logic block is not equivalent to its covered network in \textit{CPS}, the program will return ``false"; otherwise, the algorithm will continue checking until the last element in \textit{CPS} and return ``true". 
Thus, \textit{CPS} guarantees the global cover and the local block checking performs a fast verification.

\begin{example}
Fig. \ref{fig:map} (b) and (c) illustrate the cut enumeration and FPGA technology mapping steps for network \textit{OG}. After the multi-iterated cut selection, the best representative cut for each node is determined. Then, the mapping network is derived in reversed topological order, and the ordered node set is \{12,1,10,11,8,9,6,7,2,3,4,5,0\}. Based on this order vector, cut \{1,2,3,7\} is selected as the representative cut of node-12 first, and the node set \{1,10,11,8,9,6,2,3,7\} is covered, with the uncovered node set being \{7,4,5,0\}. Then, cut \{4,5\} is selected as the representative cut of node-7 and node-0 is directly connected to the output. Thus, all nodes in \textit{OG} are covered, and the mapping network \textit{MG} is derived.
As for verification on technology mapping, we first collect the injective corresponding node set \textit{CPS}, according the deriving mapping process, and which is \{\{12,7\}, \{7,6\}, \{0,0\}\}. Next, we sort \textit{CPS} in topological order, yielding \{\{0,0\},\{7,6\},\{12,7\}\}. After these steps, we can perform verification for each logic block one by one according to \textit{CPS}. Node-0 can be trivially verified because it is a constant node. If a bug is found in \textit{CPS} during the checking step, the process will stop and return ``false". Otherwise, it will return ``true" once all node-pairs in \textit{CPS} have been checked. 
\end{example}

\subsubsection{Equivalence checking on sequential mapping}
\ \newline 
\indent Sequential mapping is built upon the combinational mapping method for its combinational network between registers. It is straightforward to prove that equivalence can also be checked using equivalence checking on combinational mapping because the registers will not be altered.

\subsection{Equivalence Checking on Logic Block}
Algorithm \ref{algorithm:ec-block} shows the details for the equivalence checking on logic block. It is easy to add constraint of logic block for regular networks, but choice networks are different and the difference will be discussed in this part.
In order to distinguish it from \textit{OG}, the internal node is called block in \textit{MG}.
And we refer to the \textit{OG}'s node in $np$ as $fn$ and to the \textit{MG}'s node in $np$ as $sn$.

\begin{definition}
\label{def-regulary_blcok}
(\textit{regular block}). The \textit{regular block} is a block with one root node and its limited fanin cone, and each node in its cone must be representative node. 
\end{definition}

\begin{definition}
\label{def-choice_block}
(\textit{choice block}). The \textit{choice block} is a block with one root node and its limited fanin cone, and there are some choice nodes in its fanin cone. 
\end{definition}

\vspace{3pt}
\begin{algorithm}
\small
\caption{checking logic block}
\label{algorithm:ec-block}
\begin{algorithmic}[1]
\Require node-pair $np$, network \textit{OG} before mapping, 

         mapping result network \textit{MG},
\Ensure Equivalent or not
\State initialize $error$ = false, SAT solver $s$
\State $v_1$ $\leftarrow$ add\_constraint\_of\_logic\_block($np$, \textit{MG}, $s$)
\State $v_2$ $\leftarrow$ add\_constraint\_of\_covered\_network($np$, \textit{OG}, $s$)
\State add\_xor\_constraint($s$, $v_1$, $v_2$)
\If{$s$ is SAT}:
   \State $error$ = true
\Else
   \State $error$ = false
\EndIf
\State \Return !$error$
\end{algorithmic}
\end{algorithm}

\subsubsection{Checking for regular block}
\ \newline 
\indent As illustrated in Algorithm \ref{algorithm:ec-block}, checking equivalence for a regular block is straightforward. The variable $error$ indicates that the current logic block is not equivalent to its covered network. First, we add the constraint for the logic block of $sn$. The block's truth table can be directly obtained and transformed into Conjunctive Normal Form (CNF) using Boolean operators, so we add this CNF to the SAT solver $s$ through Tseitin transformation and return a temporary variable $v_1$ for this logic block. Next, we need to consider adding the constraint of the covered network of the logic block in \textit{OG}. 

Thanks to the fact that there are no choice classes in \textit{OG} for regular network, we can directly traverse from $fn$ to its block's leaves to collect the covered nodes. The traversed nodes are guaranteed to be the block's covered network due to the property of \textit{k}-feasible cuts. Hence, we can add constraints to the covered nodes using Tseitin transformation and return another temporary variable $v_2$ for root $sn$. Finally, we need to add the XOR constraint to $s$ between $v_1$ and $v_2$, and check the satisfiability of $s$. If $s$ is SAT, it means that some bugs have occurred here, otherwise, this logic block is deemed to be correct.

\begin{example}
As shown in Fig. \ref{fig:map}(c), node-6 of \textit{MG} is a regular block according to its covered node set \{4,5,7\} in \textit{OG}. And we can do logic block verification following Algorithm \ref{algorithm:ec-block}.
In \textit{MG}, we start by adding constraints for block-6. The truth table of block-6 can be directly obtained from the mapping network and is \{0001\}, and left for the truth table is low-bit. We need to assign the distinct sat ids to the nodes for $s$, and the node-4, node-5 and block-6 are assigned the ids {1,2,3}.  Then we can get the constraints ``$(\bar{1} + \bar{3})(\bar{2} + \bar{3})(1+2+3)$" of block-6's truth table \{0001\} using Tseitin transformation, and return the sat id of block-6. 
In \textit{OG}, we also need to assign the distinct sat ids to the covered nodes for $s$. node-4 and node-5 are the same in \textit{MG}, so their sat ids are retained, and node-7 is assigned the id {4}. We can then add the constraints ``$(\bar{1} + \bar{4})(\bar{2} + \bar{4})(1+2+4)$" for node-4, node-5 and node-7 using Tseitin transformation, and return the sat id of node-7. 
Finally, we add the xor constraint for the returned sat ids \{3,4\}, and then solve the $s$ to get the checked value ``true".
\end{example}

\subsubsection{Checking for choice block}
\ \newline 
\indent Choice networks differ from the regular networks on displayed choice class in the network. It is more difficult to add constraints to the covered network by traversing from $fn$ to the block's leaves in choice networks, as the choice node can break the property of $k$-feasible cuts, meaning that it may not be possible to reach the leaves from the root. However, the representative node saves the best cuts of the choice nodes in its choice class, allowing the fanin cone of the choice nodes can be visited by their original fanout cone. Additionally, thanks to the topological order, the covered nodes can be collected from the block's leaves to root $fn$, which may result in collecting redundant nodes, but the solver's result remains unaffected. The other steps are the same as in regular networks.

\begin{example}
As shown in Fig. 3(c), it is clear that block-7 of \textit{MG} is a choice block and its exact covered node set is \{12,1,11,6,7,2,3\} in \textit{OG}. 
In \textit{MG}, we add the constraints for block-7 first, and its truth table is \{0101,0101,0001,0101\}. We assign the distinct sat ids to nodes for $s$, and node-1, node-2, node-3, block-6 and block-7 are assigned the id \{1,2,3,4,5\}. Then we get the constraints of its truth table is ``$(\bar{1}+4+5)(\bar{1}+\bar{3}+5)(\bar{1}+\bar{2}+5)(2+3+\bar{4}+\bar{5})(1+\bar{5})$", and we return the sat id of block-7. 
In \textit{OG}, collecting the covered node set by traversing from root to leaf nodes can be costly when the choice class is large. We first collect the redundant candidate-covered nodes from leaf nodes to the root direction, and the candidate node set is \{1,2,3,7,6,11,8,9,12\}. Compared to the exact covered nodes above, \{8,9\} are redundant nodes, but this does not affect the checking result. The candidate node set \{7,6,11,8,9,12\} is assigned the sat ids \{6,7,8,9,10,11\}, and \{1,2,3\} retain their original sat ids. Then, we add the constraints of these candidate-covered nodes by Tseitin transformation until the root node-12, and return the sat id of node-12.
Finally, we add the xor constraint for the returned sat ids \{5,11\}, and then solve the $s$ to get the checked value ``true".
\end{example}

\section{MEC Framework}

In this section, we will briefly introduce the system component diagrams required for FPGA technology mapping and the parts related to MEC are the ones that we mainly focus on.

\begin{figure}[htbp]
    \centering
    \includegraphics[scale=0.305]{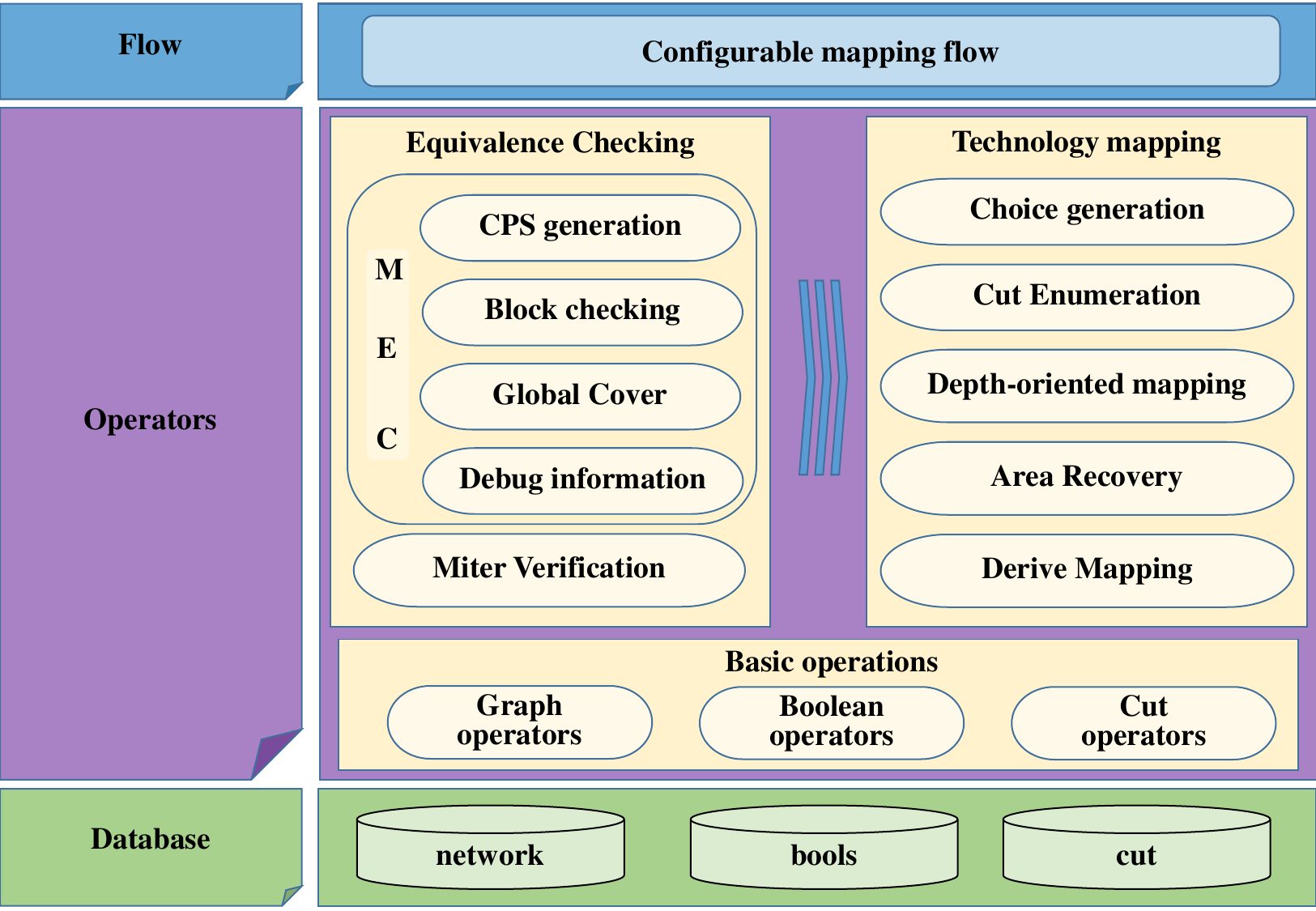}
    \caption{The framework of MEC.}
    \label{fig:framework}
\end{figure}

Fig. \ref{fig:framework} shows our system components for the technology mapping process. It consists of three layers: the database layer, operators layer and flow layer. The database layer is consisted of basic data structures, mainly about the networks, cut and Boolean logic. Therefore, it can support graph-based manipulation and cut-based operators. The bools mainly includes truth table, SOP and CNF. The upper layer is the operators layer. With the support of the basic database, we can implement the FPGA technology mapping and equivalence checking. The top layer is the flow layer. We can make the mapping flow configurable through this layer. So it is convenient for the user to do more exploration for the mapping flow.

We will focus on the technology mapping flow supported by MEC. 
The input is a network for a combinational circuit. First, We will generate the choice network by merging the selected internal optimized networks. Next, we perform cut enumeration and the following is the multi-iteration and multi-criteria cut sorting to update the representative cut for each node. In the middle of this process, we can directly call MEC to do verification by ``\textit{ICS generation}", ``\textit{Block checking}" and ``\textit{Global Cover}". If bugs are detected by MEC, the problematic part can also be extracted and saved in a DOT format file by MEC's ``\textit{Debug information}" for easier debugging. Finally, we can obtain the netlist by deriving the final mapping network. Therefore, MEC tool can support the equivalence checking at fine-grained level during technology mapping and is more efficient.

\section{Experiment}

\begin{table*}[htb]
 \renewcommand\tabcolsep{6.7pt} 
\caption{Runtime comparison between \textit{MEC}, \textit{MTR}, \textit{FML} and \textit{CEC} tools.}
\label{tab:time}
\centering

\begin{tabular}{|c|c|c|c|cc|cc|c|}
    \hline\hline
    
    \multirow{2}{*}{Type}        & \multirow{2}{*}{Name} & \multirow{2}{*}{Size} & \multirow{2}{*}{Level} & \multicolumn{2}{c|}{Coarse-grained Time (sec)} & \multicolumn{2}{c|}{Fine-grained Time (sec)}      & {Speedup}  \\ \cline{5-9} 
                                 &                       &                       &                        & \multicolumn{1}{c|}{FML} & MTR                 & \multicolumn{1}{c|}{MEC} & CEC                   & CEC/MEC \\ 
    \hline\hline
    \multirow{10}{*}{arithmetic} 
                                 & adder                 & 1277                  & 255                    & \multicolumn{1}{c|}{3.5    }  & \multicolumn{1}{c|}{0.16}     & \multicolumn{1}{c|}{\textbf{0.02}}   & 0.77      & 38.92 \\ 
                                 & bar                   & 3472                  & 12                     & \multicolumn{1}{c|}{13.13  }  & \multicolumn{1}{c|}{6.73}     & \multicolumn{1}{c|}{\textbf{0.05}}   & 1.55      & 30.97 \\
                                 & div                   & 57376                 & 4372                   & \multicolumn{1}{c|}{109,128 } & \multicolumn{1}{c|}{timeout}  & \multicolumn{1}{c|}{\textbf{49.2}}   & 417.74    & 8.49  \\
                                 & log2                  & 32093                 & 32060                  & \multicolumn{1}{c|}{37,859.1} & \multicolumn{1}{c|}{timeout}  & \multicolumn{1}{c|}{\textbf{45.91}}  & 105.9     & 2.31  \\ 
                                 & max                   & 3378                  & 287                    & \multicolumn{1}{c|}{29.74  }  & \multicolumn{1}{c|}{0.51}     & \multicolumn{1}{c|}{\textbf{0.27}}   & 1.67      & 6.08  \\ 
                                 & multiplier            & 27191                 & 274                    & \multicolumn{1}{c|}{546.44 }  & \multicolumn{1}{c|}{timeout}  & \multicolumn{1}{c|}{\textbf{3.32}}   & 30.4      & 9.14  \\
                                 & sin                   & 5441                  & 225                    & \multicolumn{1}{c|}{167.17 }  & \multicolumn{1}{c|}{timeout}  & \multicolumn{1}{c|}{\textbf{0.7}}    & 11.39     & 16.33 \\
                                 & sqrt                  & 24747                 & 5058                   & \multicolumn{1}{c|}{timeout}  & \multicolumn{1}{c|}{timeout}  & \multicolumn{1}{c|}{\textbf{25.7}}   & 153.37    & 5.97  \\ 
                                 & square                & 18549                 & 250                    & \multicolumn{1}{c|}{180.33 }  & \multicolumn{1}{c|}{timeout}  & \multicolumn{1}{c|}{\textbf{0.58}}   & 51.61     & 88.09 \\ 
                                 & hyp                   & 214592                & 24801                  & \multicolumn{1}{c|}{timeout}  & \multicolumn{1}{c|}{timeout}  & \multicolumn{1}{c|}{\textbf{473.53}} & 11331.81  & 23.93 \\ \hline
    \multirow{10}{*}{random}    
                                 & priority              & 1107                  & 250                    & \multicolumn{1}{c|}{3.93  }    & \multicolumn{1}{c|}{0.099}    & \multicolumn{1}{c|}{\textbf{0.013}}  & 0.13      & 9.71  \\ 
                                 & cavlc                 & 704                   & 16                     & \multicolumn{1}{c|}{3.69  }    & \multicolumn{1}{c|}{0.022}    & \multicolumn{1}{c|}{\textbf{0.0057}} & 0.08      & 14.04 \\ 
                                 & arbiter               & 12096                 & 87                     & \multicolumn{1}{c|}{30.36 }    & \multicolumn{1}{c|}{12.48}    & \multicolumn{1}{c|}{\textbf{0.157}}  & 4.95      & 31.51 \\ 
                                 & i2c                   & 1490                  & 20                     & \multicolumn{1}{c|}{4.33  }    & \multicolumn{1}{c|}{0.053}    & \multicolumn{1}{c|}{\textbf{0.0115}} & 0.22      & 19.00 \\ 
                                 & ctrl                  & 182                   & 10                     & \multicolumn{1}{c|}{1.67  }    & \multicolumn{1}{c|}{0.0014}   & \multicolumn{1}{c|}{\textbf{0.0013}} & 0.01      & 7.86  \\ 
                                 & dec                   & 313                   & 3                      & \multicolumn{1}{c|}{1.52  }    & \multicolumn{1}{c|}{0.0066}   & \multicolumn{1}{c|}{\textbf{0.0024}} & 0.04      & 16.37 \\ 
                                 & int2float             & 272                   & 16                     & \multicolumn{1}{c|}{2.19  }    & \multicolumn{1}{c|}{0.0039}   & \multicolumn{1}{c|}{\textbf{0.002}}  & 0.02      & 10.10 \\ 
                                 & mem\_ctrl             & 48041                 & 114                    & \multicolumn{1}{c|}{451.79}    & \multicolumn{1}{c|}{12,487.1} & \multicolumn{1}{c|}{\textbf{4.52}}   & 43.09     & 9.53  \\ 
                                 & router                & 318                   & 54                     & \multicolumn{1}{c|}{1.8   }    & \multicolumn{1}{c|}{0.003}    & \multicolumn{1}{c|}{\textbf{0.0025}} & 0.01      & 4.02  \\ 
                                 & voter                 & 14760                 & 70                     & \multicolumn{1}{c|}{45.12 }    & \multicolumn{1}{c|}{200.13}   & \multicolumn{1}{c|}{\textbf{0.46}}   & 15.61     & 33.58 \\ \hline
    \multirow{1}{*}{Average}
                                 &{-}                    &{-}                    &{-}                     & -                              & -                             & \multicolumn{1}{c|}{\textbf{30.2}}   & 608       & 20.13 \\ \hline\hline

\end{tabular}
\end{table*}

\subsection{Setup}
We implement MEC tool in C++ based on the data structure of lsils\cite{epfl-tool}, and we used minisat\cite{minisat-tool} as the SAT solver. All procedures run on an Intel(R) Xeon(R) Platinum 8260 CPU with 2.40GHz, 24 cores and 128GB RAM. 
We perform experiments on EPFL benchmark\cite{benchmark}. This benchmark contains 10 arithmetic circuits and 10 control/random circuits with circuit sizes ranging from 0.1k to 210k.

\subsection{Results}
For convenience, \textit{MEC} refers to our proposed algorithm; 
\textit{MTR} refers to purely perform verification on the miter of the mapping network and its original network, then checking by SAT solver without any tricks;
\textit{FML} refers to directly perform verification on these two networks by the Synopsys Formality Q-2019.12-SP1;
\textit{CEC} refers to perform verification on these two networks by the cec command of abc, and we transform the gate-level netlist into aig file and the original as its inputs.
We set 100 hours as the timeout for this comparison.

As shown in Table \ref{tab:time}, it depicts that the coarse-grained methods \textit{FML} and \textit{MTR}'s performance are poor, there are even some timeout cases of arithmetic circuits. However, it also shows that the fine-grained methods \textit{MEC} and \textit{CEC} can also solve all the circuits within timeout constraints. 
\textit{MEC} achieves the best average time result and can solve most cases within 1 second; \textit{CEC} also performs good due to its exploration of the potential function-equivalent node-pairs to reduce the solver space.


We also compare the speed-up between \textit{MEC} and \textit{CEC}. 
MEC can be up to 88x faster than \textit{CEC} and the average speed-up is 20x. 
These results show that MEC is very efficient for verification on technology mapping, which would greatly improve the efficiency of verification.

\section{Conclusion}
In this paper, we proposed an efficient algorithm for verification on technology mapping, and we release an open-source FPGA mapping verification tool, MEC.
The experimental results show that MEC can solve all the cases and get most results within seconds, outperforming the miter-based method and even the commercial formality tool. In the future, our focus will be directed towards the following areas: (1) Accelerating MEC through parallelization. Our verification algorithm for technology mapping is based on local logic blocks, which has a lot of room for parallelization; (2) Our proposed algorithm has the potential for performing verification between the network in technology mapping and place-and-route (PnR) steps, as long as the behavior of the registers remains unchanged.

\section*{Acknowledgment}
This work is supported in part by the Major Key Project of PCL (No. PCL2021A08), the National Natural Science Foundation of China (No. 62090024), the National Natural Science Foundation of China (No. 62090021).

\bibliographystyle{IEEEtran}

\bibliography{formal.bib}  

\end{CJK*}
\end{document}